\documentclass[10pt,journal]{IEEEtran}
\IEEEoverridecommandlockouts

\usepackage{cite}
\usepackage{amsmath}
\usepackage{amsthm}
\usepackage{amssymb}
\DeclareMathOperator*{\argmax}{arg\,max}
\DeclareMathOperator*{\argmin}{arg\,min}
\usepackage{graphicx}
\usepackage{textcomp}
\usepackage{xcolor}
\usepackage{tabularray}
\usepackage{booktabs}
\usepackage{multirow}
\usepackage{makecell}

\usepackage{algorithm}
\usepackage[noend]{algpseudocode}
\makeatletter
\def\BState{\State\hskip-\ALG@thistlm}
\makeatother
\newtheorem{proposition}{Proposition}
\newtheorem{definition}{Definition}
\newtheorem{remark}{Remark}

\usepackage{matlab-prettifier}
\usepackage{color}
\usepackage{mathrsfs}
\usepackage[shortlabels]{enumitem}
\makeatletter
\def\BState{\State\hskip-\ALG@thistlm}
\makeatother
\def\BibTeX{{\rm B\kern-.05em{\sc i\kern-.025em b}\kern-.08em
    T\kern-.1667em\lower.7ex\hbox{E}\kern-.125emX}}
\usepackage{subcaption}

\begin{document}

\title{{\bf\Large PRADA: Proactive Risk Assessment and Mitigation of Misinformed Demand Attacks on Navigational Route Recommendations}
}

\author{Ya-Ting Yang, Haozhe~Lei, and Quanyan Zhu
\thanks{The Authors are with the Department of Electrical and Computer Engineering, New York University, Brooklyn, NY, 11201, USA; E-mail: {\tt\small \{yy4348, hl4155, qz494\}@nyu.edu}. Correspondence should be sent to YY.}%
\thanks{This work has been submitted to the IEEE for possible publication. Copyright may be transferred without notice, after which this version may no longer be accessible.}
}

\maketitle

\begin{abstract}
Leveraging recent advances in wireless communication, IoT, and AI, intelligent transportation systems (ITS) played an important role in reducing traffic congestion and enhancing user experience. Within ITS, navigational recommendation systems (NRS) are essential for helping users simplify route choices in urban environments. However, NRS are vulnerable to information-based attacks that can manipulate both the NRS and users to achieve the objectives of the malicious entities.
This study aims to assess the risks of misinformed demand attacks, where attackers use techniques like Sybil-based attacks to manipulate the demands of certain origins and destinations considered by the NRS. We propose a game-theoretic framework for proactive risk assessment of demand attacks (PRADA) and treat the interaction between attackers and the NRS as a Stackelberg game. 
Specifically, we consider the case of local-targeted attacks, in which the attacker aims to make the NRS recommend the authentic users towards a specific road that favors certain groups. Our analysis unveils the equivalence between users' incentive compatibility and Wardrop equilibrium recommendations and shows that the NRS and its users are at high risk when encountering intelligent attackers who can significantly alter user routes by strategically fabricating non-existent demands.
To mitigate these risks, we introduce a trust mechanism that leverages users' confidence in the integrity of the NRS, and show that it can effectively reduce the impact of misinformed demand attacks. Numerical experiments are used to corroborate the results and demonstrate a Resilience Paradox, where locally targeted attacks can sometimes benefit the overall traffic conditions.
\end{abstract}

\begin{IEEEkeywords}
Information attack, risk assessment, Stackelberg game, navigational recommendation.
\end{IEEEkeywords}

\section{Introduction}

Harnessing vast information available from modern wireless communication and Internet of Things (IoT) advancements \cite{lv2020ai,zantalis2019review}, coupled with the progress made in data science and artificial intelligence \cite{veres2019deep,haydari2020deep}, intelligent transportation systems (ITS) have gained substantial attention for their ability to effectively tackle traffic congestion and elevate driver experiences. Within ITS, navigational recommendation systems (NRS) such as Google Maps and Apple Maps play a vital role in complex urban environments, as users, including drivers and pedestrians, may be overwhelmed by diverse route choices \cite{van2016user}. Based on the given information, the NRS offers routes to simplify users' decision-making processes, aiming to reduce travel duration, elevate user experiences, and alleviate congestion \cite{8262884}. However, unlike routing in computer network systems \cite{Multihoming,zhang2010optimizing} or routing in transportation networks for connected autonomous vehicles \cite{rossi2018routing}, the NRS involves human drivers who may not always adhere to the recommendations, making user compliance unguaranteed \cite{10.1145/3474717.3483652}. Therefore, in this work, we refer to the NRS as the platform that provides incentive-compatible recommendations \cite{ning2023robust,yang2023strategic}. This ensures that users can not be better off by unilaterally deviating from the recommended strategies, and can be interpreted as a routing game between NRS users.

\begin{figure}
    \centering
    \includegraphics[width=3.3in]{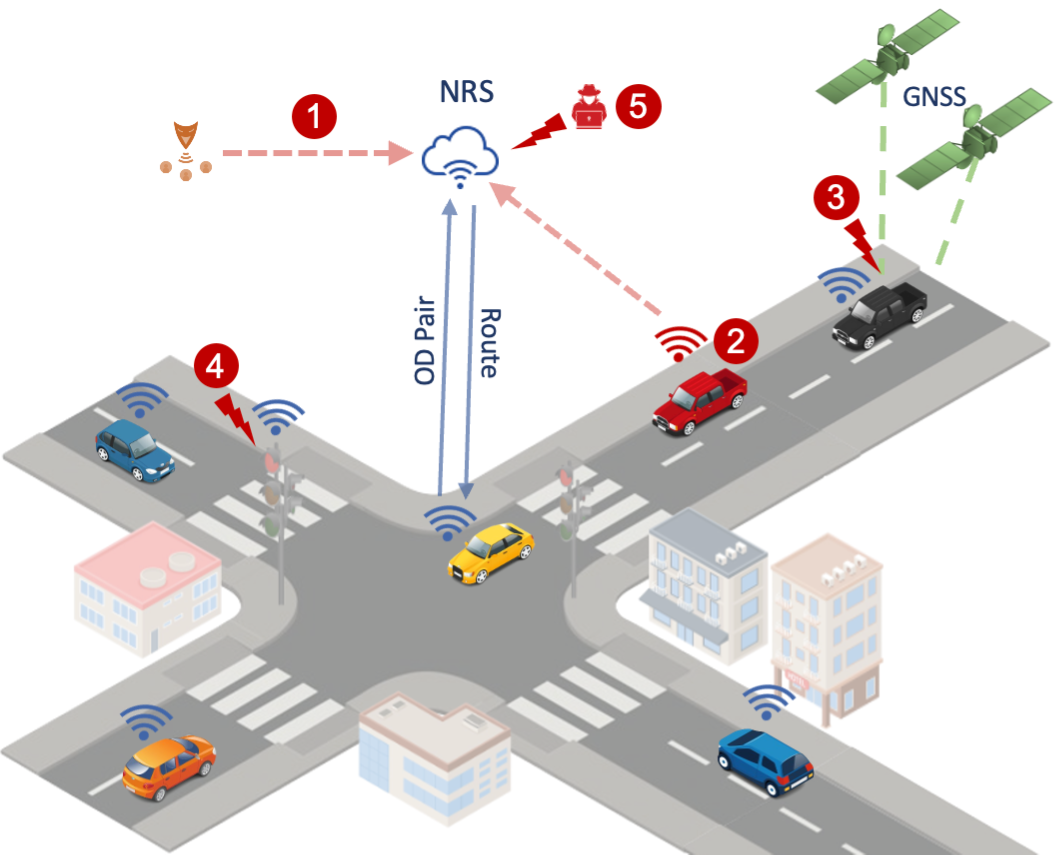}
    \caption{The NRS receives user origin and destination (OD) requests, which are vulnerable to exploitation by malicious entities for demand attacks: (1) Sybil attack (2) Botnet (3) GPS spoofing (4) Man-in-the-middle attack (5) Insider threats.}
    \label{fig:NRS_process}
\vspace{-3mm}
\end{figure}

Nonetheless, as illustrated in Fig. \ref{fig:NRS_process}, the navigational recommendation process is prone to various vulnerabilities \cite{mecheva2020cybersecurity} that attackers can leverage to promote particular groups or businesses in a locally targeted sense or potentially exacerbate congestion levels on a broader network-wide scale. Within this context, information-based attacks emerge as a critical concern, as they empower malicious entities to spread misinformation and manipulate both the NRS and its users to achieve their objectives \cite{Waniek2021}. For example, on the Waze platform, residents may fabricate congestion reports to divert traffic away from their residential areas, aiming to maintain the tranquility of their surroundings \cite{waze_resident}. Additionally, the recent study \cite{eryonucu2022sybil} demonstrates how Sybil-based attacks can effectively manipulate perceived crowdedness at places of interest and traffic congestion levels within Google Maps, which serves as an example that shows how fake users (that can lead to fake demands) can impact the NRS. These recent findings highlight that the informational attack is a significant threat within NRS.  Therefore, it is essential to understand the impact of these informational attacks, evaluate their risks, and develop strategies to mitigate and manage them. Within this scope, our study is motivated by the recent research \cite{eryonucu2022sybil} and specifically focuses on the misinformed \textit{demand attack}, which is defined as the manipulation of user demands between certain origins and destinations, and can be achieved through a variety of attack methods, including but not limited to Sybil-based attack, botnet, GPS spoofing, man-in-the-middle attack on wireless communication, insider threat, etc.

To proactively assess the risks of demand attacks on the NRS, one viable approach is to employ a data-driven method. This involves constructing or utilizing an existing NRS to simulate users over an urban transportation network, then launching demand manipulation attacks through various methods to collect data for analysis \cite{sarker2020cybersecurity}. While the data-driven method can yield realistic results, it is often task-specific and time-consuming to implement.
A more cost-effective alternative is the model-based approach, which can also provide us with an outcome that is capable of measuring and analyzing
the risk. However, the model-based approach is typically fragmented, as comprehensive risk assessment requires multiple models, including those for attackers, the NRS, users, and their interactions.
To bridge this gap, we introduce a holistic model-based approach using a game-theoretic framework for proactive risk assessment of demand attacks (PRADA). As depicted in Fig. \ref{fig:dara}, this approach aims to integrate necessary models into a cohesive framework, providing more comprehensive results and analyses for risk assessment.

The proposed PRADA framework is analyzed through three layers of games. The \textit{inner layer} is the model for the NRS, as incentive-compatible recommendations can be interpreted as a routing game between NRS users. The equivalence between user incentive-compatible recommendations and Wardrop equilibrium recommendations aids in analyzing the \textit{middle layer}, which models the interaction between the threat model of the malicious entity and the NRS as a Stackelberg game \cite{fang2021introduction,yang2023strategic}. In this game, the attacker acts as the leader who conveys misinformed demands, while the NRS, as the follower, responds to the provided information. The \textit{outer layer} captures the interplay between the PRADA risk evaluator (responsible for conducting the risk assessment) and the Stackelberg game at the middle layer. Specifically, this study focuses on the locally targeted attack for the attacker's objective, as it has rather few systematic studies. In this attack, the attacker manipulates the NRS to direct genuine users towards a specific road that benefits certain groups or businesses. From the perspective of the PRADA risk evaluator, our analysis demonstrates that by strategically designing the misinformed demands, such as how many fake users for which OD pairs, the attacker can make the NRS redistribute the true users originally on other alternative paths towards the target road, leading to a higher risk for the NRS and its users.

With the proposed PRADA framework, this study takes one step further by introducing a viable mitigation method. We utilize the concept of \textit{trust score}, which measures how much users trust the integrity of the NRS and believe the recommendations are not manipulated. A higher trust score indicates greater user willingness to follow recommendations that differ from previous ones for the same origin and destination. By proposing a trust mechanism that incorporates trust score constraints into the model for NRS, we can effectively reduce the impact of demand attacks and lower the risk both locally and network-wide. Our analysis indicates that the dual variable associated with the trust constraint can be interpreted as a \textit{trust risk factor}. It shows how sensitive the user's expected cost is to changes in the trust score.

To this end, our contribution can be summarized as follows:
\begin{itemize}
    \item We identify vulnerabilities of the NRS and then propose a game-theoretic framework for holistic proactive risk assessment for misinformed demand attacks (PRADA). 
    \item We employ a Stackelberg game approach to integrate the threat model with the NRS into the PRADA framework. Our analytical results and numerical studies demonstrate that users are at high risk when encountering intelligent attackers who target specific roads by fabricating fake demands on alternative paths.
    \item We introduce a trust mechanism that leverages users' confidence in the integrity of the NRS. Our findings show that the resulting trusted recommendation can effectively mitigate the impact of demand attacks, both in local-targeted and network-wide contexts.
\end{itemize}

\section{Literature Review}

\noindent \textbf{Generic Attacks on ITS:} Intending to enhance mobility, safety, sustainability, and traffic efficiency in urban transportation networks, modern ITS leverages a wide range of advanced technologies. These include sensors and cameras for data collection, wireless communication—particularly vehicle-to-everything (V2X) technology—for information exchange, GPS for precise positioning, data analytics coupled with AI for processing, as well as mobile apps for distributing information. However, the extensive network of interconnected devices with vast information exchanged in ITS introduces vulnerabilities \cite{hahn2019security} that expand the potential cyber-physical attack surface (see \cite{huq2017cyberattacks} for real-world ITS attack cases). Within the domain of ITS, malicious entities or potential adversaries \cite{mecheva2020cybersecurity,cybersecurity_transit} can exploit vulnerabilities within data and information infrastructure through a class of attacks known as informational attacks. These attacks, which include data falsification, data integrity breaches, and data poisoning \cite{almalki2020review,pan2022poisoned}, are designed to divert drivers and escalate traffic congestion that leads to increased crash risks within urban transportation networks. These attacks exploit various tactics, including sensor and GPS spoofing techniques \cite{8761439}, as well as employing man-in-the-middle and Sybil-based methods \cite{8337099,eryonucu2022sybil}.

\noindent \textbf{Informational Attacks on NRS:} We scrutinize the particular vulnerabilities inherent to the NRS, which are susceptible to a wide range of potential attacks \cite{8761439}. Regarding generic attacks in the scenario of navigational guidance, attackers could compromise vehicles via wireless communication networks or manipulate real-time traffic conditions, leading to informational attacks. Such attacks can result in inaccurate traffic predictions and misguidance for drivers, contributing to network-wide traffic congestion and safety concerns \cite{8375813}. For more specific examples, \cite{10.1145/3032970.3032983} illustrates how the availability of portable GPS signal spoofing devices enables attackers to divert drivers from their intended destinations without their awareness. Additionally, \cite{8337099} demonstrates the significant impact of a single Sybil device with limited resources on platforms like Waze, where false reports of congestion and accidents can automatically reroute user traffic. This work expands the threats discovered by recent studies \cite{eryonucu2022sybil} targeting NRS, where misinformation, such as fabricated demands, originates from Sybil-based users.
Specifically, we assess the risk of locally targeted attacks that have rather few systematic studies, wherein attackers tend to strategically mislead users onto specific roads that favor certain groups. 

\noindent \textbf{Risk Assessment:}
Risk assessment is a systematic process for identifying, analyzing, and evaluating risks within a particular system or framework in various domains, including but not limited to energy systems \cite{chehri2021security}, supply chains \cite{pournader2020review}, IoT-based systems \cite{ntafloukas2022cyber}, autonomous vehicles \cite{li2021risk}, and transportation networks \cite{koohathongsumrit2021integrated,luo2021threat}. It aims to understand potential adverse outcomes, enabling organizations or individuals to mitigate or manage such risks effectively. 
Within the field of transportation, risk assessment plays an important role, as evidenced by the substantial focus on (highway) crash risk evaluation \cite{liu2022transfer}, collision risk avoidance for autonomous vehicles \cite{li2021risk}, and risk-based route selection \cite{koohathongsumrit2021integrated}. 
When addressing potential cyber risks in ITS, a deeper understanding of the attack model is necessary \cite{kalinin2021cybersecurity}. Cyber attackers are often intelligent and strategic, unlike non-strategic attackers who add disturbances uniformly or randomly. Therefore, a natural way to integrate the attack model into risk assessment and capture the interaction between the attacker and the target system is through game-theoretic approaches. These approaches are commonly employed to capture the threat posed by followers in dynamic games, such as Stackelberg games \cite{casorran2019study,yang2023game}, bargaining games \cite{guerrero2018solving}, as well as in mechanism design problems involving contract designs \cite{zhang2019mathtt} and incentive mechanisms \cite{zhu2012guidex}, offering analytical tools and strategies for effective risk assessment and mitigation.

\begin{figure}
    \centering
    \includegraphics[width=3.4in]{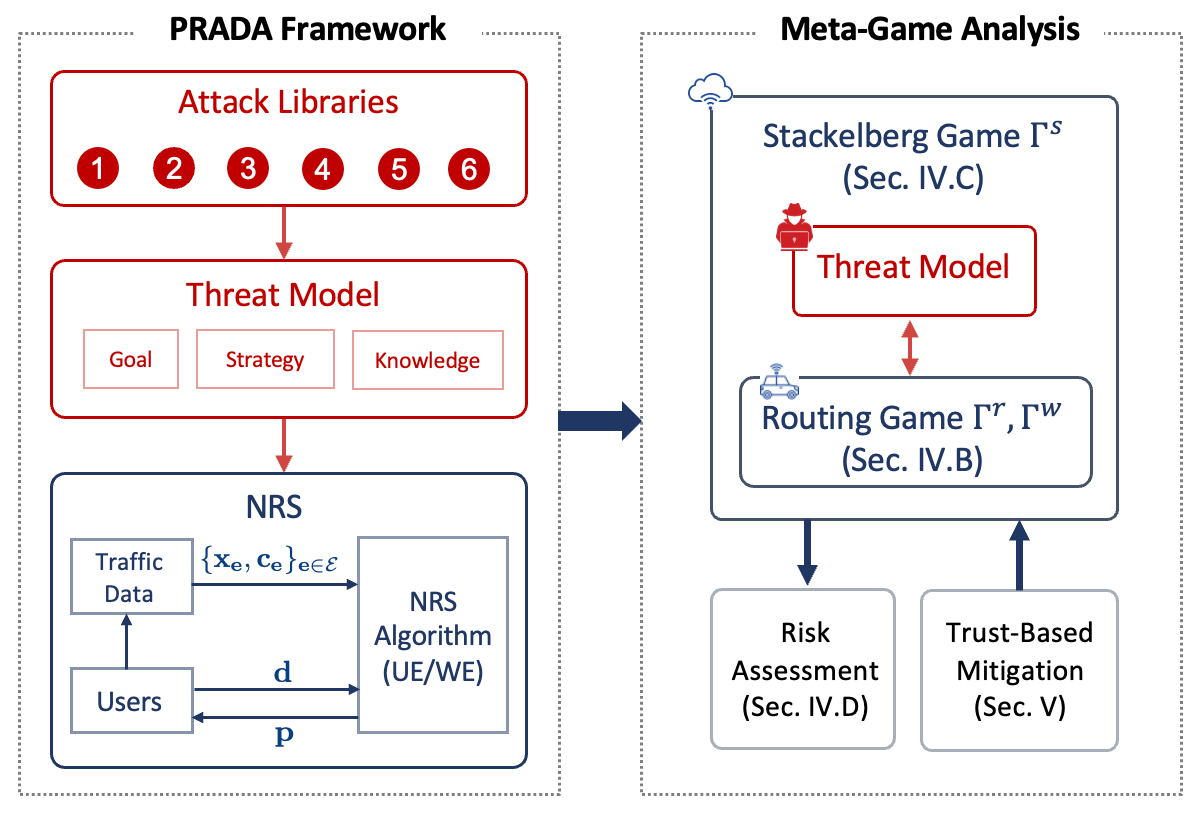}
    \caption{The PRADA framework is analyzed through three layers of games between users, the NRS, the threat model of the attacker, and the PRADA risk evaluator.}
    \label{fig:dara}
\vspace{-4mm}
\end{figure}

\section{Misinformed Demand Attack}
The nature of ITS, characterized by a vast network of interconnected devices and extensive data exchange, presents vulnerabilities that can be exploited by malicious entities through informational attacks, which target the system's data and information infrastructure. In this study, we specifically investigate one type of informational attack within the context of navigational recommendations, called the \textit{demand attack}. As depicted in Fig. \ref{fig:NRS_process}, the NRS typically receives navigational requests, including origin and desired destination (OD) pairs, from users. These requests contribute to the demand associated with each OD pair. More specifically, we denote $\Theta$ as the set of all possible OD pairs, and each OD pair $\theta \in \Theta$ is associated with a demand $d_\theta \in \mathbb{R}_{\ge 0}$ contributed from the users. Let $\textbf{d} = \{d_\theta\}_{\theta \in \Theta}$ represent all the demands for later usage. For instance, suppose there are ten NRS users who wish to travel from the Empire State Building (origin) to Times Square (destination), which corresponds to OD pair $\theta$, then the demand for this OD pair is $d_\theta=10$. 

Different demands generally lead to different recommendations from the NRS. For instance, with a single user, the NRS can simply suggest the shortest path. However, if many users are associated with the same OD pair, recommending the shortest path can lead to the `flash crowd effect' \cite{flash_crowd_effect}, making it no longer the optimal choice. Malicious entities can exploit this fact to manipulate demand $\textbf{d}$ to some other $\textbf{d}^\prime$, steering the NRS toward other recommendations that fulfill their own objectives. Therefore, the PRADA risk evaluator, who is responsible for conducting the risk assessment, must evaluate the risks for various libraries of attacks, consisting of different malicious goals,  types of attackers, and attack methods that can lead to demand attacks.

\subsection{Demand Attack Methods} \label{sec:att_methods}
In this subsection, we mention some techniques indicated in Fig. \ref{fig:NRS_process} that the attacker with related knowledge can utilize to launch the misinformed demand attack.

\subsubsection{Sybil Attacks}
Attackers can generate numerous fake identities (non-existent users) as shown in Fig. \ref{fig:NRS_process} and then simulate these users at specific locations \cite{8337099}. These non-existent users send OD pair requests to the NRS through emulators \cite{eryonucu2022sybil}. When computing recommendations, the NRS considers these fake demands alongside genuine ones, leading to recommendations that differ from those based solely on authentic demands. Consequently, the attacker can strategically redistribute legitimate users by launching Sybil-based attacks with fake demands. Furthermore, since the non-existent users do not actually drive on the roads after receiving the recommendations, the actual traffic conditions caused by legitimate users will differ from the NRS's expectations.

\subsubsection{Botnet}
An attacker can deploy a botnet, a network of compromised devices controlled remotely \cite{ashraf2021iotbot}. These devices can range from infected smartphones to IoT devices and computers. The attacker can command the botnet to send numerous navigational OD pair requests to the NRS, as shown in Fig. \ref{fig:NRS_process}, simulating authentic users seeking guidance between various origins and destinations. Similar to the Sybil attack, the resulting recommendations will differ due to the fake demands (that do not exist on roads) generated by the botnet. This allows the attacker to  redistribute legitimate users by utilizing the botnet to strategically create fake demands, aligning the recommendations with the attacker's objectives.

\subsubsection{GPS Spoofing}
Attackers can employ GPS spoofing techniques \cite{10.1145/3032970.3032983,8761439} to alter the perceived location of NRS users, as illustrated in Fig. \ref{fig:NRS_process}. This manipulation can cause users to send navigation requests with incorrect origins. For example, when a user selects `current location' as the origin, the spoofed GPS signal can make the NRS believe the user is in a different place. Consequently, the requested OD pair is being manipulated, affecting the overall demand considered by the NRS. This disruption can significantly impact navigational recommendations, especially if multiple NRS users are affected simultaneously, leading to incorrect navigational recommendations 
 (that may align with the attacker's objective).

\subsubsection{Wireless Communication Network}
We use the man-in-the-middle attack as an illustrative example of a demand attack through wireless communication networks \cite{al2020review}. In this scenario, attackers attempt to position themselves between the user's device and the NRS server, allowing them to intercept communications in between (see Fig. \ref{fig:NRS_process}).
When a user sends a request for route recommendations for a specific OD pair, the attacker modifies the request before it reaches the system's servers. This modification can involve altering the origin, destination, or other parameters within the request. By manipulating multiple requests from different users, the attacker can increase or decrease the demand for specific OD pairs. As a result, the manipulated demand influences the NRS's recommendation to align with the attacker's objectives. 

\subsubsection{Insider Threats}
An insider, such as an employee or contractor with access to the NRS infrastructure, can directly manipulate the data or algorithm within the system. This manipulation may involve altering the demand data, such as increasing or decreasing the number of requests for specific OD pairs. By manipulating the demand data, the insider can bias the recommendations suggested by the NRS to align with their objectives.

\subsection{Types of Attackers}
We can categorize attackers into two main types.
\subsubsection{Non-Strategic Attacker} Non-strategic attackers may lack the understanding of how the NRS generates recommendations for users, or they may not pay attention to and simply disregard this information. Instead, a non-strategic attacker often manipulates demands by uniformly or randomly increasing or decreasing the number of requests associated with some OD pairs. They then observe whether such manipulation achieves their desired outcome.

\subsubsection{Strategic Attacker}
Strategic attackers are often more intelligent and possess either a deep understanding of how the NRS generates recommendations for users or the ability to model the NRS. With this knowledge, they assess or observe the outcomes of the NRS when under attack. By leveraging this insight, strategic attackers can utilize efficient strategies to achieve their objectives with fewer resources used and less time spent.

\subsection{Attacker's Goals}
Imagine an attacker driven by self-interest, in conflict with the overall social welfare goal of reducing congestion. This scenario can be studied at both local targeted and network-wide levels: the former pertains to specific groups or locations, while the latter considers the system-wide impact.

\subsubsection{Local-Targeted Attacks}The attacker seeks to bias the system by suggesting paths that favor particular groups (e.g., higher-paying users) or businesses (e.g., those paying the attacker to ensure users see particular ads or pass by their shops). For example, a restaurant owner could pay malicious entities to ensure a certain volume of users are directed by the NRS to pass by the road where the restaurant is located, as illustrated in Fig. \ref{fig:target}.
\begin{figure}
    \centering
    \includegraphics[width=3.0in]{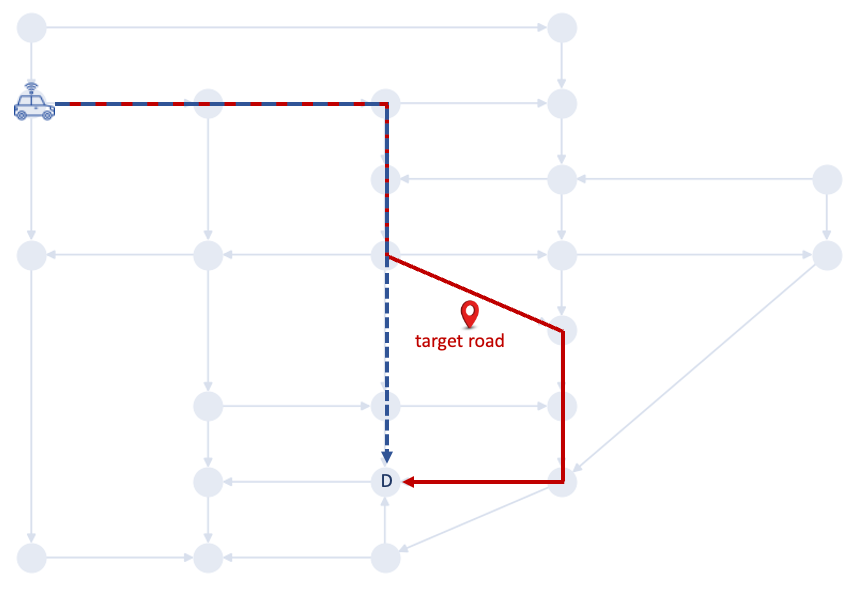}
    \caption{Local-targeted attacks: The malicious entity manipulates the NRS to guide users through the target road. The blue dashed line represents the original recommendation from the NRS, and the red line illustrates the one under attack.}
    \label{fig:target}
\vspace{-3mm}
\end{figure}
The impact of these attacks can be measured by the difference in traffic flow on specific roads with and without the attacks.

\subsubsection{Network-Wide Attacks} The attacker aims to disrupt the system by increasing delays or congestion levels across the network, consequently raising the overall travel time cost for users. These attacks can harm the system's reputation, leading to user dissatisfaction or a loss of trust in the NRS. 

It is worth noting that while this study focuses on malicious entities with local-targeted objectives, these attacks can result in both local-targeted and network-wide impacts. The metrics used to assess these risks are detailed in Section \ref{sec:metrics}.

\section{Proactive Risk Assessment}
This section aims to assess the risk caused by misinformed demand attacks on the NRS. The proposed framework for proactive risk assessment of demand attacks (PRADA) is illustrated in Fig. \ref{fig:dara}. 
The PRADA evaluator proactively evaluates risks by employing a library of attack scenarios. Each attack scenario is characterized by the attacker's goal, attack strategy (type of attacker), and attack method (knowledge), which together form the threat model. The resulting demand attack from the threat model then affects the recommendations suggested by the NRS.

The framework is analyzed through meta-game, which consists of three layers of games.
The inner layer focuses on the routing game between NRS users. That is, the NRS aims to provide incentive-compatible recommendations to users, as human users may not always follow recommendations if they find better alternatives that align with their preferences. 
The middle layer employs a Stackelberg game approach to capture the interaction between the threat model of the attacker and the NRS. Here, the attacker acts as the leader, manipulating demand, while the NRS, as the follower, responds to the provided information.
The outer layer is a meta-game for the interplay between the PRADA risk evaluator and the middle layer. It involves assessing the impacts of different threat models from the attack libraries.
We then conduct a sensitivity analysis of demand attacks and propose metrics for measuring local-targeted and network-wide impacts.

\subsection{Settings for Urban Transportation Networks}
The urban transportation network can be represented by $\mathcal{G}=\{\mathcal{V}, \mathcal{E}\}$, where the set of nodes $\mathcal{V}$ corresponds to intersections and the set of edges $\mathcal{E}$ indicates the roads. Traveling along a road $e\in \mathcal{E}$ incurs a road-specific cost $c_e: \mathbb{R}_{\ge 0} \mapsto \mathbb{R}_{+}$ associated with the flow $x_e \in \mathbb{R}_{\ge 0}$ on road $e$. One usual choice for the cost $c_e(\cdot)$ is the standard Bureau of Public Roads (BPR) function $$c_e(x_e)=t_e\left(1+\alpha\left(\frac{x_e}{k_e}\right)^\beta\right)$$ for travel time costs. Here, $t_e \in \mathbb{R}_{+}$ represents the free-flow travel time on road $e$, $k_e \in \mathbb{R}_{+}$ signifies the capacity of road $e$, and $\alpha, \beta \in \mathbb{R}_{\ge 0}$ are some parameters.

\subsection{Models for NRS}
The user set of NRS is denoted as $\mathcal{U}$. Each user, $u \in \mathcal{U}$, is associated with a specific origin $O_u \in \mathcal{V}$ and destination $D_u \in \mathcal{V}$ pair. We refer to the pair as an OD pair, expressed by $\theta_u = (O_u, D_u)$, and the set of OD pairs for the NRS users is $\Theta_{\mathcal{U}} \subseteq \Theta$, with $\Theta =  |\mathcal{V}| \times |\mathcal{V}|$ representing all the possible OD pairs within the network. Then, user $u$ with OD pair $\theta_u$ has the feasible path choice set $\mathcal{S}_u  = \{s_{u, 1}, \ldots, s_{u, k_u}\}$, which is identical to the feasible path choice set $\mathcal{S}_\theta  = \{s_{\theta, 1}, \ldots, s_{\theta, k_{\theta}}\}$ for OD pair $\theta$, where $\theta=\theta_u$. Each choice $s_{u, i} \in \mathcal{S}_u$ or $s_{\theta, i} \in \mathcal{S}_\theta$ provides the user $u$ a path from the origin to the desired destination. To this end, the elements of the urban transportation network considered by the NRS can be encapsulated using the notation $\mathscr{R}=\left\langle   \mathcal{G}, (c_e(\cdot))_{e \in \mathcal{E}}, \mathcal{U}, (\mathcal{S}_u)_{u \in \mathcal{U}}  \right\rangle$, and we call $\mathscr{R}$ the ``NRS component''.  

\subsubsection{User Equilibrium Recommendations (UE)}

We consider the scenario where the NRS recommends a mixed strategy over feasible path choices to the users.
Define $\mathcal{P}_u:=\Delta \mathcal{S}_u$ as the simplex of $\mathcal{S}_u$. A mixed strategy for user $u$ is $\textbf{p}_u \in \mathcal{P}_u$ so that $\textbf{p}_u=\{p_{u, i}\}_{s_{u, i} \in \mathcal{S}_u}$ is a probability distribution over $\mathcal{S}_u$. Each element $p_{u, i}\in [0,1]$ denotes the probability that the NRS recommends path $s_{u, i} \in \mathcal{S}_u$ to user $u$, and needs to satisfy the constraints $\sum_{i=1}^{k_u}p_{u, i}=1 , \ \forall u \in \mathcal{U}$. That is, 
$$\mathcal{P}_u := \left\{\textbf{p}_{u} \in \mathbb{R}^{k_u} \middle| p_{u,i}\geq 0, i=1, \cdots, k_u , \sum^{k_u}_{i=1}p_{u,i}=1 \right\}.
$$ Then, let $\mathcal{P}:=\Pi_{u \in \mathcal{U}}\mathcal{P}_u$, the recommendation suggested by the NRS to all users is $\textbf{p}=\{\textbf{p}_u\}_{u \in \mathcal{U}} \in \mathcal{P}$.

In transportation, from a microscopic perspective, the probability $p_{u, i}$ can be interpreted as the expected volume generated by user $u$ along path $s_{u, i}$. This, in turn, contributes to the expected road flow (load) $x_e^r: \mathcal{P} \mapsto \mathbb{R}_{\ge 0}$ on road $e \in \mathcal{E}$ as below. 
\begin{equation*}
    x_e^r(\textbf{p})=\sum_{u \in \mathcal{U}}\sum_{s_{u, i} \in \mathcal{S}_u}p_{u, i}a_{es_{u, i}},
\end{equation*} where $a_{es_{u, i}}$ is an element of the road-path incidence matrix $A_{|\mathcal{E}|\times |\Pi_{u \in \mathcal{U}} \mathcal{S}_u|}=[a_{es_{u, i}}]$, and is defined as follows.
$$ a_{es_{u, i}}=
\begin{cases}
    1 \qquad \text{if} \ e \in s_{u, i},\\
    0 \qquad \text{otherwise}.
\end{cases}
$$ Hence, a generalized travel cost $C_{u, i}: \mathcal{P} \mapsto \mathbb{R}_{+}$ for path $s_{u, i}$ can be formulated by summing the costs of all the roads along the path:
\begin{equation*}
    C_{u, i}(\textbf{p})=\sum_{e \in s_{u, i}}c_e(x_e^r(\textbf{p})).
\end{equation*} In this context, the expected cost evaluated by user $u$ is $F_u^r: \mathcal{P} \mapsto \mathbb{R}_{\ge 0}$, where
\begin{equation}
    F_u^r(\textbf{p}_{u},\textbf{p}_{-u}) = \sum^{k_u}_{i=1}p_{u,i}C_{u,i}(\textbf{p}_{u},\textbf{p}_{-u}).
\label{eq:F_u}
\end{equation}
Note that a recommendation $\textbf{p} \in \mathcal{P}$ can be interpreted as the strategy profile in a routing game between NRS users. Hence, the routing game addressed by the NRS can be defined as $\Gamma^r=\langle \mathscr{R}, \mathscr{F}^r \rangle$, where $\mathscr{F}^r=(F_u^r)_{u \in \mathcal{U}}$ represents the costs evaluated by users.
However, human users may choose not to follow the NRS recommendation if they find a better alternative. Therefore, to ensure user adherence that leads to a guaranteed performance over the network, the NRS must suggest a recommendation \(\textbf{p} \in \mathcal{P}\), where \(\textbf{p}_u \in \mathcal{P}_u\) is preferred by user \(u\) given the recommendations to other users \(\textbf{p}_{-u} \in \Pi_{u^{\prime} \in \mathcal{U} \setminus \{u\}} \mathcal{P}_{u^{\prime}}\), for all \(u \in \mathcal{U}\). That is, given the recommendations \(\textbf{p}_{-u}\) to users other than \(u\), user \(u\) has no incentive to unilaterally deviate from the recommended \(\textbf{p}_u\). This coincides with the concept of user equilibrium (UE), which is defined as follows:
\begin{definition}[User Equilibrium Recommendation]
    Considering a routing game addressed by the NRS defined as $\Gamma^r=\langle\mathscr{R}, \mathscr{F}^r\rangle$, a mixed strategy profile $\textbf{p} \in \mathcal{P}$ for all the users is called a user equilibrium recommendation if it satisfies:
    \begin{align}
    F_u^r(\textbf{p}_{u},\textbf{p}_{-u})-F_u^r(\textbf{p}^{\prime}_{u},\textbf{p}_{-u}) \leq 0 , \forall \ \textbf{p}^{\prime}_{u} \in \mathcal{P}_u, \forall u \in \mathcal{U}.
    \label{prob:RS}
    \end{align}
\label{def:NRS}
\vspace{-3mm}
\end{definition}
UE recommendation can be found by gradient descent-based method. Let $\text{proj}_{\mathcal{P}_u}$ represent the projection onto simplex $\mathcal{P}_u$ and $\rho \in \mathbb{R}$ denote the step size, problem \eqref{prob:RS} can be solved by finding a fixed point to:
\begin{equation}
    \textbf{P}_u^{*} = \text{proj}_{\mathcal{P}_u} \left[\textbf{p}_u^{*} - \rho \nabla_u F_u^r(\textbf{p}_u^{*}, \textbf{p}_{-u}^{*})\right], \ \forall u \in \mathcal{U}.
\label{eq:PGD}
\end{equation}

\subsubsection{Wardrop Equilibrium Recommendations (WE)}
In this subsection, we use the concept of Wardrop equilibrium (WE) as the foundation for the WE-based recommendations.

Recall that $\mathcal{U}$ represents the user set of the NRS, with their associated set of OD pairs, denoted as $\Theta_{\mathcal{U}}$. For each OD pair $\theta \in \Theta_{\mathcal{U}}$, the demand flow aggregated from users, $d_\theta = \sum_{u \in \mathcal{U}} \mathbf{1}_{\{\theta_u = \theta\}}$, must be routed from the corresponding origin to the desired destination. As for OD pair $\theta \in \Theta \setminus \Theta_{\mathcal{U}}, d_\theta = 0$. The set of feasible paths for each OD pair $\theta$ is $\mathcal{S}_\theta = \{s_{\theta, 1}, \ldots, s_{k_\theta}\}$. Then, let vector $\textbf{y}_\theta=\{y_{\theta, i}\}_{s_{\theta, i} \in \mathcal{S}_\theta} \in \mathbb{R}^{k_\theta}$ so that each element $y_{\theta, i}$ represents the expected flow generated by the users being recommended through path $s_{\theta, i}$, and needs to satisfy the constraints:
$\sum_{s_{\theta, i} \in \mathcal{S}_\theta} y_{\theta, i} = d_\theta$. That is, we can define
$$\mathcal{Y}_\theta := \left\{\textbf{y}_{\theta} \in \mathbb{R}^{k_\theta} \middle| y_{\theta,i}\geq 0, i=1, \cdots, k_\theta , \sum^{k_\theta}_{i=1}y_{\theta,i}=d_\theta \right\}.
$$ By denoting $\mathcal{Y}:= \Pi_{\theta \in \Theta} \mathcal{Y}_\theta$, the expected flow recommended by the NRS on all the paths $s_{\theta, i} \in \mathcal{S}_\theta, \theta \in \Theta$ is $\textbf{y}=\{\textbf{y}_\theta\}_{\theta \in \Theta} \in \mathcal{Y}$. Similar to the UE recommendation, $\textbf{y}$ also contributes to the expected road flow (load) $x_e^w: \mathcal{Y} \mapsto \mathbb{R}_{\ge 0}$ on road $e \in \mathcal{E}$ as below. 
\begin{equation*}
    x_e^w(\textbf{y})=\sum_{\theta \in \Theta}\sum_{s_{\theta, i} \in \mathcal{S}_\theta}y_{\theta, i} a^\prime_{es_{\theta, i}},
\end{equation*} where $a^\prime_{es_{\theta, i}}$ is an element of the road-path incidence matrix $A^\prime_{|\mathcal{E}|\times |\Pi_{\theta \in \Theta} \mathcal{S}_\theta|}=[a^\prime_{es_{\theta, i}}]$, and is defined as follows.
$$ a^\prime_{es_{\theta, i}}=
\begin{cases}
    1 \qquad \text{if} \ e \in s_{\theta, i},\\
    0 \qquad \text{otherwise}.
\end{cases}
$$ Then, the cost evaluated by user $u$ with OD pair $\theta=\theta_u$ for path $s_{\theta, i}$ is $F_{\theta, i}^w: \mathcal{Y} \mapsto \mathbb{R}_{\ge 0}$, where
\begin{equation*}
    F_{\theta, i}^w(\textbf{y}) = \sum_{e \in s_{\theta, i}} c_e\left(x_e^w(\textbf{y})\right).
\end{equation*}

In this context, the routing game addressed by the NRS can be defined as $\Gamma^w=\langle \mathscr{R}, \mathscr{F}^w \rangle$, where $\mathscr{F}^w=(F_{\theta, i}^w)_{s_{\theta, i} \in \mathcal{S}_\theta, \theta \in \Theta}$ represents the costs evaluated by users. Then, the WE-based recommendation is defined as the following.

\begin{definition}[Wardrop Equilibrium Recommendation]
    Consider a routing game addressed by the NRS, defined as $\Gamma^w=\langle\mathscr{R}, \mathscr{F}^w\rangle$.
    A feasible path flow and road load pair $(\textbf{y}, \textbf{x}^w)$ with $\textbf{y} \in \mathcal{Y}$ and $\textbf{x}^w = \{x_e^w(\textbf{y})\}_{e \in \mathcal{E}} \in \mathbb{R}_{\ge 0}^{|\mathcal{E}|}$ is called a Wardrop equilibrium recommendation if it satisfies: 
\begin{equation}
    \begin{aligned}
        F_{\theta, i}^w(\textbf{y}) \leq F_{\theta, j}^w(\textbf{y}) \ \text{when} \ \ y_{\theta, i}>0 &,\\
        \qquad \qquad \forall s_{\theta, i}, s_{\theta, j} \in \mathcal{S}_\theta, \ \forall \theta \in \Theta &.
    \end{aligned}
\end{equation}
\label{def:WE}
\vspace{-3mm}
\end{definition} In other words, the WE-based recommendation results in the minimal prevailing costs for all used paths. Then, according to Beckmann \cite{beckmann1956studies}, WE can be computed as the solution to the following optimization problem,
\begin{subequations}
    \begin{align}
        W(\textbf{d}): \ \min_{\textbf{y}, \textbf{x}^w} \ &\sum_{e \in \mathcal{E}} \int_{0}^{x_e^w} c_e(z) dz\\
        \text{s.t. } \ & \sum^{k_\theta}_{i=1}y_{\theta,i}=d_\theta, \forall \theta \in \Theta, \\
        & y_{\theta,i} \geq 0, \forall s_{\theta, i} \in \mathcal{S}_\theta, \forall \theta \in \Theta,\\
        & x_e^w(\textbf{y})=\sum_{\theta \in \Theta}\sum_{s_{\theta, i} \in \mathcal{S}_\theta}y_{\theta, i} a^\prime_{es_{\theta, i}}, \forall e \in \mathcal{E},
    \end{align}
\label{prob:WE_xy}
\end{subequations} and the corresponding WE recommendation pair is represented as $(\widehat{\textbf{y}}, \widehat{\textbf{x}}^w)$.

\subsubsection{Connection between UE and WE}

Correspondence can be observed between Definition \ref{def:NRS} and \ref{def:WE}, which assists us in the subsequent analysis when integrating the demand attack model into our PRADA framework. That is, user $u$ and the set of feasible paths $\mathcal{S}_u$ in UE correspond to OD pair $\theta$ and $\mathcal{S}_\theta$ in WE, the expected flow $x_e^r$ on the road $e$ correspond to the road load $x_e^w$, and the probability $p_{u, i}$ that user $u$ be recommended on the path $s_{u, i}$ in UE corresponds to the path flow $y_{\theta, i}$ in WE. Note that letting $d_{\theta}=1$ in WE leads to $\sum^{k_\theta}_{i=1}y_{\theta,i}=1$ that corresponds to $\textbf{p}_u \in \mathcal{P}_u$ for $\theta = \theta_u$.

Let ($\textbf{p}$, $\textbf{x}^r$) denote the UE pair, where $\textbf{p} \in \mathcal{P}$ and $\textbf{x}^r = \{x_e^r(\textbf{p})\}_{e \in \mathcal{E}} \in \mathbb{R}_{\ge 0}^{|\mathcal{E}|}$. As a result, the WE solution pair $(\widehat{\textbf{y}}, \widehat{\textbf{x}}^w)$ that corresponds to the ($\textbf{p}$, $\textbf{x}^r$) pair can be viewed as a feasible solution for the UE recommendation. More specifically, if ($\textbf{p}$, $\textbf{x}^r$) is a WE, then for all user $u \in \mathcal{U}$ in constraint \eqref{prob:RS}:
Since only paths with minimum cost are utilized, all the paths used by any given user have the same cost. That is, for $p_{u, i} > 0$, the cost $C_{u,i}(\textbf{p}_{u},\textbf{p}_{-u})$ should be the same for $u$. The overall expected cost $\sum^{k_u}_{i=1}p_{u,i}C_{u,i}(\textbf{p}_{u},\textbf{p}_{-u})$ is independent of the probability $p_{u, i}$ of $C_{u,i}(\textbf{p}_{u},\textbf{p}_{-u})$. Lastly, note that if ($\textbf{p}$, $\textbf{x}^r$) is an equilibrium, there is no incentive for a user $u$ to deviate to any other $\textbf{p}^\prime_u \in \mathcal{P}_u$.

\begin{proposition}
    A WE solution pair $(\widehat{\textbf{y}}, \widehat{\textbf{x}}^w)$ defined in Definition \ref{def:WE} that corresponds to the UE pair ($\textbf{p}$, $\textbf{x}^r$) is a feasible solution for the UE recommendation defined in Definition \ref{def:NRS}.
\end{proposition}
\begin{proof}
The proof follows from the discussion above.
\end{proof}

Similarly, in the case where $d_{\theta}=\sum_{u \in \mathcal{U}} \mathbf{1}_{\{\theta_u=\theta\}} > 1$, the NRS can recommend a mixed strategy $p_{u, i}$ over the feasible path $s_{u, i} \in \mathcal{S}_u=\mathcal{S}_{\theta_u}$, with each $p_{u, i}=\widehat{y}_{\theta, i}/d_{\theta}$, where $\widehat{y}_{\theta, i}$ comes from the WE flow-load pair $(\widehat{\textbf{y}}, \widehat{\textbf{x}}^w)$ in which only paths with minimum cost are utilized for each OD pair. Following this recommendation, the recommended path flow $\textbf{y}^r:=\{y^r_{\theta, i}\}_{\theta \in \Theta, s_{\theta, i} \in \mathcal{S}_{\theta}}$, with each $y^r_{\theta, i} = \sum_{u \in \mathcal{U}} p_{u, i} \mathbf{1}_{\{\theta_u=\theta\}}$, is the same as the WE path flow $\widehat{\textbf{y}}$.

Since a feasible UE recommendation can be computed from the WE solution pair$(\widehat{\textbf{y}}, \widehat{\textbf{x}}^w)$, we have the following remarks.
\begin{remark}
    The recommended road flow load $\textbf{x}^r$ provided by the UE recommendation coincides with $\widehat{\textbf{x}}^w$ in WE.
\end{remark}
It is worth noting that, according to problem \eqref{prob:WE_xy}, the WE road flow load $\widehat{\textbf{x}}^w$ is uniquely determined if the cost function $c_e(\cdot)$ on each road $e \in \mathcal{E}$ is strictly increasing in the corresponding road flow load $x_e$ \cite{still2018lectures}.
\begin{remark}
    Under the assumption that the cost function $c_e(\cdot)$ on each road $e \in \mathcal{E}$ is strictly increasing in the corresponding road flow load $x_e$, the UE road flow load $\textbf{x}^r$ is also uniquely determined and is equivalent to $\widehat{\textbf{x}}^w$. Hence, we can denote $\textbf{x} \in \mathbb{R}^{|\mathcal{E}|}_{\ge 0}$ and have $\textbf{x} = \widehat{\textbf{x}}^w = \textbf{x}^r$.
\label{remark:x_eq}
\end{remark}

\subsection{Stackelberg Game for Risk Assessment} \label{sec:md_problem}
To assess the risks of different attack methods and types, we adopt a Stackelberg game approach to capture the interaction between the attacker (AT) and the NRS. In this setup, the attacker acts as the leader, deciding on the misinformed demands, while the NRS, as the follower, generates recommendations based on the provided information. 

\subsubsection{Strategic Attackers with Local-targeted Objectives}
The subsequent analyses focus on scenarios involving strategic attackers with local-targeted attack objectives. Specifically, the attacker intends to have a desired level of expected flow load caused by genuine NRS users on the target road $e^{\prime} \in \mathcal{E}$. That is, the attacker aims to make $x_{e^\prime}^r(\textbf{p}) = \sum_{u \in \mathcal{U}}\sum_{s_{u, i} \in \mathcal{S}_u}p_{u, i}a_{e^\prime s_{u, i}}$ in UE recommendation, which is equivalent to $x^w_{e^\prime}(\widehat{\textbf{y}})=\sum_{\theta \in \Theta}\sum_{s_{\theta, i} \in \mathcal{S}_\theta}\widehat{y}_{\theta, i}a^\prime_{e^\prime s_{\theta, i}}$ in WE according to Remark \ref{remark:x_eq}, achieve a desired level $\gamma \in \mathbb{R}_{\ge 0}$.


Consider the situation where a Sybil-based attacker generates non-existent demands $\textbf{d}^a \in \mathcal{D}$, where $\mathcal{D}=\mathbb{Z}_{\ge 0}^{|\Theta|}$ using Sybil (fake) users. Then, the NRS will need to consider an aggregated demand of $\textbf{d}^{\prime}=\textbf{d} + \textbf{d}^a$ when generating the recommendation. Note that for each OD pair $\theta$, the demand $d_\theta^{\prime}$ under attack consists of $d_\theta + d_\theta^a$. Without loss of generality, we can assume that only a proportion of $\frac{d_\theta}{d_\theta + d_\theta^a}$ of the WE expected path flow $\widehat{y}_{\theta, i}^{\prime}$ with respect to $W(\textbf{d}^{\prime})$ is caused by authentic users. Hence, we denote $\widehat{y}^u_{\theta, i}=\left(\frac{d_\theta}{d_\theta + d_\theta^a}\right)\widehat{y}_{\theta, i}^{\prime}$ and $\widehat{\textbf{y}}^u = \{\widehat{y}^u_{\theta, i}\}_{\theta \in \Theta, s_{\theta, i} \in \mathcal{S}_{\theta}}$ for path flow generated by true users. In this case, the attacker aims to make the flow load from genuine users $x^w_{e^\prime}(\widehat{\textbf{y}}^u) = \sum_{\theta \in \Theta}\sum_{s_{\theta, i} \in \mathcal{S}_\theta}\widehat{y}^u_{\theta, i}a^\prime_{e^\prime s_{\theta, i}}$ on the targeted road $e^\prime$ reach the desired level $\gamma$. To this end, the Stackelberg game between the attacker (AT) and the NRS can be defined as $\Gamma^s=\langle \text{AT}, \mathcal{D}, U_{AT}, (e^\prime, \gamma), \Gamma^w \rangle$, where $U_{AT}: \mathcal{D} \mapsto \mathbb{R}_{\ge 0}$ is the attacker's cost in terms of the resources spent in fabricating fake demands.  The leader-follower problem is formulated as follows, and can be solved by gradient descent-based algorithms \cite{boyd2004convex}.
\begin{subequations}
    \begin{align}
        \min_{\textbf{d}^a} \ & U_{AT}(\textbf{d}^a)=\sum_{\theta \in\Theta} d_\theta^a\\
        \text{s.t.} \ &\sum_{\theta \in \Theta}\sum_{s_{\theta, i} \in \mathcal{S}_\theta}\widehat{y}^u_{\theta, i}a^\prime_{e^\prime s_{\theta, i}}\geq \gamma, \\
        &(\widehat{\textbf{y}}, \widehat{\textbf{x}}^w) \in \argmin W(\textbf{d} + \textbf{d}^a), \\
        & d_\theta^a \geq 0, \forall \theta \in \Theta.
    \end{align}
\label{prob:RS_attack}
\end{subequations}
Lastly, note that the level $\gamma$ in problem \eqref{prob:RS_attack} can not be arbitrarily large, which leads us to the following.
\begin{remark}
    The edge load $x_e^w$ is bounded by the total demand of the true user $d_T=\sum_{\theta \in \Theta}d_\theta$. Hence, the attacker's desired level $\gamma$ is also upper-bounded by $d_T=\sum_{\theta \in \Theta}d_\theta$.
\end{remark}

\subsubsection{Sensitivity Analysis for Demand Attack} \label{sec:sensitivity}
Under the assumption that the cost function $c_e(\cdot)$ is continuous and increasing in $x_e$, a pair $(\textbf{y}, \textbf{x})$ is a minimizer of
$W(\textbf{d})$ if and only if it satisfies the following Karush–Kuhn–Tucker (KKT) conditions.
\begin{subequations}
    \begin{align}
        c_e(x_e)-\lambda_e =0, \ &\forall e \in \mathcal{E},\label{eq:c_lambda}\\
        -\nu_\theta + \sum_{e \in \mathcal{E}}\lambda_e a^\prime_{es_{\theta, i}} - \mu_{\theta, i}=0, \ &\forall s_{\theta, i} \in \mathcal{S}_\theta, \forall \theta \in \Theta, \label{eq:nu_lambda_mu}\\
        \mu_{\theta, i}y_{\theta, i}=0,\ &\forall s_{\theta, i} \in \mathcal{S}_\theta, \forall \theta \in \Theta \label{eq:mu_y},
    \end{align}
\label{eq:KKT}
\end{subequations} with Lagrangian multipliers $\nu_\theta \in \mathbb{R}_{\ge 0}, \forall \theta \in \Theta$, $\lambda_e \in \mathbb{R}_{\ge 0}, \forall e \in \mathcal{E}$, and $\mu_{\theta, i} \in \mathbb{R}_{\ge 0}, \forall s_{\theta, i} \in \mathcal{S}_\theta, \forall \theta \in \Theta$. Then, a pair $(\textbf{y}, \textbf{x})$ satisfying the constraints with multipliers $-\boldsymbol{\nu}=-(\nu_\theta)_{\theta\in \Theta}, \boldsymbol{\lambda}=(\lambda_e)_{e \in \mathcal{E}}, \boldsymbol{\mu}=(\mu_{\theta, i})_{s_{\theta, i} \in \mathcal{S}_\theta, \theta \in \Theta}$ also satisfies $$\nu_\theta=\sum_{e \in s_{\theta, i}}c_e(x_e)-\mu_{\theta, i}
\begin{cases}
    =\sum_{e \in s_{\theta, i}}c_e(x_e), \ y_{\theta, i} > 0, \\
    \leq \sum_{e \in s_{\theta, i}}c_e(x_e), \ y_{\theta, i} = 0, 
\end{cases}$$ which coincides with the definition of WE recommendation.


Then, with KKT conditions, we aim to examine how the WE pair $(\widehat{\textbf{y}}, \widehat{\textbf{x}}^w)$ can be influenced by changes in the demand $\textbf{d}$ according to \cite{still2018lectures}.

\vspace{-1mm}
\begin{proposition}
    Let the pair $(\textbf{y}, \textbf{x})$ with the corresponding multipliers $\boldsymbol{\nu}$ and $\boldsymbol{\mu}$ described in \eqref{eq:KKT} be a WE for demand $\textbf{d}$ and $(\textbf{y}', \textbf{x}')$ with corresponding multipliers $\boldsymbol{\nu}'$ and $\boldsymbol{\mu}'$ be a WE for demand $\textbf{d}'$. Then, $(\boldsymbol{\nu}'-\boldsymbol{\nu})^{T}(\textbf{d}'-\textbf{d}) \geq \boldsymbol{\mu}'^{T}\textbf{y}+\boldsymbol{\mu}^{T}\textbf{y}' \geq 0$.
\label{prop:diff_d}
\end{proposition}
 
\begin{proof}
    Under the assumption that the cost function $c_e(\cdot)$ is continuous and increasing in $x_e$, which indicates that $[c_e(x^\prime_e)-c_e(x_e)](x^\prime_e-x_e) \geq 0, \forall e \in \mathcal{E}$, then with $x_e^\prime = x_e^w(\textbf{y}^\prime)=\sum_{\theta \in \Theta}\sum_{s_{\theta, i} \in \mathcal{S}_\theta}y_{\theta, i}^\prime a^\prime_{es_{\theta, i}}$ and $x_e = x_e^w(\textbf{y})=\sum_{\theta \in \Theta}\sum_{s_{\theta, i} \in \mathcal{S}_\theta}y_{\theta, i} a^\prime_{es_{\theta, i}}$, we have the following: $$\sum_{\theta \in \Theta}\sum_{s_{\theta, i} \in \mathcal{S}_\theta} \sum_{e \in \mathcal{E}} [c_e(x^\prime_e)-c_e(x_e)] a^\prime_{es_{\theta, i}}(y_{\theta, i}^\prime-y_{\theta, i}) \geq 0.$$ Note that the KKT conditions in \eqref{eq:c_lambda} give us $c_e(x_e) = \lambda_e, \forall e \in \mathcal{E}$. With \eqref{eq:nu_lambda_mu}, the above inequality becomes: 
    $$\sum_{\theta \in \Theta}\sum_{s_{\theta, i} \in \mathcal{S}_\theta} [(\nu_\theta^\prime+\mu_{\theta, i}^\prime) - (\nu_\theta+\mu_{\theta, i})] (y_{\theta, i}^\prime-y_{\theta, i}) \geq 0.$$ By \eqref{eq:mu_y} and $\sum_{s_{\theta, i}\in \mathcal{S}_\theta}y_{\theta,i}=d_\theta$, we have the following:
    $$\sum_{\theta \in \Theta} (\nu_\theta^\prime-\nu_\theta)(d_{\theta}^\prime-d_{\theta}) \geq \sum_{\theta \in \Theta}\sum_{s_{\theta, i} \in \mathcal{S}_\theta} (\mu_{\theta, i}^\prime y_{\theta, i}+\mu_{\theta, i} y_{\theta, i}^\prime).$$ Note that $\mu_{\theta, i}^\prime, y_{\theta, i}, \mu_{\theta, i}, y_{\theta, i}^\prime \geq 0$, we complete the proof.
\end{proof} 
The result of Proposition \ref{prop:diff_d} can also be written as
$$\bigg[\sum_{e \in s_{\theta, i}}c_e(x_e')-c_e(x_e)\bigg](d_\theta'-d_\theta)\geq 0, \forall s_{\theta, i} \in \mathcal{S}_\theta$$ with $y_{\theta, i}', y_{\theta, i}>0$.
The Proposition \ref{prop:diff_d} states that if one demand $d_\theta$ is increased by fake users, with other demands remaining the same, then the equilibrium cost $\nu_\theta$ perceived by the NRS for the user $u$ with OD pair $\theta_u=\theta$ is also increased. 
\begin{proposition}
    For $W(\textbf{d})$ with demand $\textbf{d}$, let $\textbf{x}$ be a WE corresponds to cost $c_e(\cdot)$ and $\textbf{x}'$ be a WE corresponds to cost $c_e'(\cdot)$, then $\left[c_e'(x_e)-c_e(x_e)\right](x_e'-x_e) \leq 0$ and $\left[c_e'(x_e')-c_e(x_e')\right](x_e'-x_e) \leq 0$.
\label{prop:diff_c}
\begin{proof}
    Under the assumption that the cost functions $c_e(\cdot)$ and $c_e^\prime(\cdot)$ are continuous and increasing in $x_e$ and $x_e^\prime$, respectively, with the optimality conditions for $\textbf{x}^\prime$ be a WE corresponds to cost function $c_e^\prime(\cdot)$ and $\textbf{x}$ be a WE corresponds to $c_e(\cdot)$, we have the following inequalities:
    \begin{subequations}
        \begin{align}
            \left[c_e(x_e^\prime)-c_e(x_e)\right](x_e^\prime-x_e) &\geq 0 \label{eq:mono_c}\\
            \left[c_e^\prime(x_e^\prime)-c_e^\prime(x_e)\right](x_e^\prime-x_e) &\geq 0 \label{eq:mono_c_prime}\\
            c_e(x_e)(x_e^\prime - x_e) &\geq 0\label{eq:opt_}\\
            c_e^\prime(x_e^\prime)(x_e - x_e^\prime) &\geq 0 \label{eq:opt_prime}
        \end{align}
    \end{subequations}
    Then, summing up \eqref{eq:mono_c}, \eqref{eq:opt_}, and \eqref{eq:opt_prime} gives us $\left[c_e(x_e')-c_e'(x_e')\right](x_e'-x_e) \geq 0$, while adding \eqref{eq:mono_c_prime}, \eqref{eq:opt_}, and \eqref{eq:opt_prime} leads us to $\left[c_e(x_e)-c_e'(x_e)\right](x_e'-x_e) \geq 0$. We complete the proof.
\end{proof}
\end{proposition}
That is, Proposition \ref{prop:diff_c} demonstrates that an increasing cost on a road $e \in \mathcal{E}$ will cause the equilibrium load $x_e$ on that road to decrease. This reduced load can be interpreted as a redistribution to alternative feasible paths. Thus, the attacker can achieve the desired flow load level $\gamma$ on the target road $e^\prime$ by redistributing the load there. Alongside Proposition \ref{prop:diff_d}, this can be accomplished by strategically increasing the perceived demands (by adding non-existent ones) on certain roads, thereby raising their costs and leading to redistribution.

\subsection{Impact Metrics for Risk Reports} \label{sec:metrics}
In this subsection, we introduce two metrics as the outcomes of our PRADA framework. 
Let $\textbf{p}$ be the recommendation to all the users without attack and $\textbf{p}^\prime$ is the one under attack.

\subsubsection{Local-targeted Impact}
We define the \textit{targeted impact metric (TI)} as the difference in traffic flow on each road with and without the demand attack, divided by the flow without the attack. Specifically, for each road $e \in \mathcal{E}$, the measure $TI_{e}$ is given by:
\begin{equation}
TI_e = \frac{|x_e^r(\textbf{p}^\prime) - x_e^r(\textbf{p})|}{x_e^r(\textbf{p})},  
    \label{eq:ti}
\end{equation} which can assists in measuring the percentage change in traffic flow on specific or targeted road affected by the demand attack. A larger value of $TI_e$ indicates the road $e$ is influenced more, often implying higher risk under the attack.

\subsubsection{Network-wide Impact}
Given the metrics $TI_{e}, \forall e \in \mathcal{E}$, we define the \textit{network impact metric (NI)} as the mean of $TI_{e}$ across all the roads/edges within the network. The measure $NI$ is as follows:
\begin{equation}
NI = \frac{1}{|\mathcal{E}|}\sum_{e \in \mathcal{E}} TI_{e}=\frac{1}{|\mathcal{E}|}\sum_{e \in \mathcal{E}} \frac{|x_e^r(\textbf{p}^\prime) - x_e^r(\textbf{p})|}{x_e^r(\textbf{p})}. 
    \label{eq:ni}
\end{equation} The metric $NI$ allows us to evaluate the percentage change in traffic flow across the entire network. It is important to note that if the demand attack primarily affects traffic flow on roads within a small area, as indicated by the $TI_{e}$ values for those roads, this localized impact will be averaged out when considering the network-wide impact.

\section{Risk Mitigation}
In this section, we aim to explore an effective mechanism to mitigate the impact of misinformed demand attacks. 

\subsection{Mitigation through User Trust}

Denote $\textbf{p}_u^o = \{p_{u, i}^o\}_{s_{u, i} \in \mathcal{S}_u} \in \mathcal{P}_u$ as the UE recommendation defined in Definition \ref{def:NRS} for user $u$ without attack (under normal traffic condition), which can be obtained from previous experience of requesting recommendation for the same OD pair. Let \( T_u \in \mathbb{R}_{\ge 0} \) for all \( u \in \mathcal{U} \) represent the \textit{trust score}, which quantifies the degree to which users trust that some malicious entities do not manipulate the recommendations provided by the NRS. The trust score then leads to the following trust constraint, implying that the user trusts the recommendation $\textbf{p}_u \in \mathcal{P}_u$ because its deviation from the user's previous experience is within an acceptable range (i.e., $T_u$). 
\begin{definition}[Trust Constraint (TC)]
    Consider an NRS component, denoted as $\mathscr{R}$. A recommended mixed strategy $\textbf{p}_u \in \mathcal{P}_u$ for a user $u$ is said to satisfy the trust constraint if the distance (in terms of Kullback–Leibler (KL) divergence) between the currently and previous recommended strategy $\textbf{p}_u^o \in \mathcal{P}_u$ is less than the trust score $T_u \in \mathbb{R}_{\ge 0}:$
    \begin{equation} D(\textbf{p}_u||\textbf{p}_u^o)=\sum^{k_u}_{i=1} p_{u,i} \log \left(\frac{p_{u,i}}{p_{u,i}^o}\right) \leq T_u.
    \label{eq:TC}
    \end{equation}
\label{def:IR}
\vspace{-3mm}
\end{definition}
A higher trust score indicates greater user trust in the current integrity of the NRS, as its associated trust constraint \eqref{eq:TC} forms a larger \textit{trust region}. This means the user is more willing to tolerate larger differences between current and previous recommendations, believing that such variations are due to changes in traffic conditions rather than malicious demand manipulation. In contrast, $T_u=0$ indicates a lack of trust in the current NRS, leading the user to follow only recommendations that match their past experiences.

Moreover, let $\textbf{p}_u$ be the recommendation to user $u$ without attack and $\textbf{p}_u^\prime$ is the one manipulated by the attacker. If each element $p_{u, i}^\prime = p_{u, i} + \epsilon_{u, i}$, where $\epsilon_{u, i}$ is a small perturbation due to demand attacks, we have the following sensitivity property for the manipulated recommendation using first-order Taylor expansion.
\begin{equation*}
    D(\textbf{p}_u^\prime||\textbf{p}_u^o) - D(\textbf{p}_u||\textbf{p}_u^o) \approx  \sum_{i=1}^{k_u} \epsilon_{u, i} \left(\log \frac{p_{u, i}}{p^o_{u, i}} + 1\right).
    \label{eq:TC_prop}
\end{equation*}
From the attacker's point of view, in order to fulfill TC, $D(\textbf{p}_u^\prime||\textbf{p}_u^o) \leq T_u$, so that user $u$ is still willing follow the manipulated recommendation $\textbf{p}_u^\prime$, the perturbations $\epsilon_{u, i}, \forall s_{u, i} \in \mathcal{S}_u$ must satisfy:
\begin{equation*}
    D(\textbf{p}_u||\textbf{p}_u^o) + \sum_{i=1}^{k_u} \epsilon_{u, i} \left(\log \frac{p_{u, i}}{p^o_{u, i}} + 1\right) \leq T_u, 
\end{equation*} which indicates that the trust score $T_u$ bounds the total perturbation to the probabilities $p_{u, i}$ on feasible paths $s_{u, i} \in \mathcal{S}_u$.

%
%

\subsection{Trust Mechanism for NRS}

Since TC bounds and mitigates the severity of manipulation caused by demand attacks, it is reasonable to incorporate such a user trust mechanism into the navigation recommendation process. In this context, users are either learned or warned to follow only the recommendations that satisfy their TC. Then, with Definition \ref{def:NRS} and \ref{def:IR}, the NRS must identify feasible recommendations that can be trusted by all users, ensuring their participation and adherence to the recommended strategies. This leads to the following definition.
\begin{definition}[Trusted Recommendation]
    Considering a routing game addressed by the NRS defined as $\Gamma^r=\langle\mathscr{R}, \mathscr{F}^r\rangle$, a trusted recommendation profile to all users $\textbf{p} \in \mathcal{P}$ needs to satisfy:
\begin{subequations}
  \begin{align}
     F_u^r(\textbf{p}_{u},\textbf{p}_{-u})-F_u^r(\textbf{p}^{\prime}_{u},\textbf{p}_{-u}) &\leq 0 , \forall \ \textbf{p}^{\prime}_{u} \in \mathcal{P}_u, \forall u \in \mathcal{U}, \label{eq:IC_cons}\\
     D(\textbf{p}_u||\textbf{p}_u^o)-T_u &\leq 0, \ \forall u \in \mathcal{U}, \label{eq:IR_cons}
  \end{align}
\label{prob:RS2}
\end{subequations} 
\label{def:NRS2}
\vspace{-5mm}
\end{definition}
Incorporating trust constraints ensures that the current recommendation provided by the NRS cannot deviate significantly from the previous one for the same OD pair, based on the assumption that traffic conditions usually evolve smoothly. Consequently, if there is a sudden change in demand caused by malicious entities, the recommendation will stay relatively aligned with past recommendations in normal circumstances.


The NRS's problem of finding trusted recommended mixed strategies for all users can also be interpreted using a non-cooperative game, defined as $\Tilde{\Gamma}^r:=(\mathscr{R}, \mathscr{F}^r, (T_u)_{u \in \mathcal{U}})$, where $\mathscr{F}^r=(F_u^r)_{u \in \mathcal{U}}$ and $F_u^r$ is expressed in \eqref{eq:F_u}. Each user $u \in \mathcal{U}$ of the NRS is a player of the game $\Tilde{\Gamma}^r$. User $u$ aims to minimize his/her own expected cost $F_u^r$ by deciding a mixed strategy $\textbf{p}_u \in \mathcal{P}_u$ over feasible path choice set $\mathcal{S}_u$ given other users' strategies $\textbf{p}_{-u}$, under the trust constraint that $\textbf{p}_u$ cannot deviate too much from previous experience $\textbf{p}_u^o$. That is, for all user $u \in \mathcal{U}$ in $\Tilde{\Gamma}^r$, given other users' strategies $\textbf{p}_{-u}$, 
\begin{equation}
\begin{aligned}
    \text{OP}_u: \min_{\textbf{p}_{u} \in \mathcal{P}_u} \ &F_u^r(\textbf{p}_{u}, \textbf{p}_{-u})\\
    \text{s.t. } 
    & D(\textbf{p}_u||\textbf{p}_u^o))-T_u \leq 0.
\end{aligned}
\label{eq:OP_u} 
\end{equation}

Then, by denoting $C^\prime_{u, i}(\textbf{p})=\sum_{e \in s_{u, i}} t_e \big[1+\alpha\big(\frac{x_e^r(\textbf{p})}{k_e}\big)^\beta + \beta x_e^u(\textbf{p}_u) \frac{\alpha}{k_e}\left(\frac{x_e^r(\textbf{p})}{k_e}\right)^{\beta-1}\big]$, we have the following proposition.
\begin{proposition}
    Consider the problem defined in \eqref{prob:RS2}. Under the conditions that for all $u \in \mathcal{U}$, the expected cost $F_u$ is continuously differentiable in $\textbf{p} \in \mathcal{P}$ and convex in $\textbf{p}_u \in \mathcal{P}_u$, and that the trust constraint is active, the trusted recommendation $\textbf{p}^*$ is as follows: $\forall u \in \mathcal{U}$,
    \begin{equation}
    p_{u, i}^* =  \frac{p_{u, i}^o \exp\left(\frac{-C^\prime_{u, i}(\textbf{p})}{\lambda_u}-1\right)}{\mu_{u, i}^\prime}, \forall s_{u, i} \in \mathcal{S}_u,
\label{eq:p_miti}
\end{equation} where $\mu_{u}^\prime =\sum_{i=1}^{k_u} p_{u, i}^*$ is the normalization term and $\lambda_u \in \mathbb{R}_{+}$ is the Lagrange multiplier for the trust constraint.
\label{prop:miti}
\vspace{-3mm}
\end{proposition}

\begin{proof}
    Let $\mathscr{L}_u (\textbf{p}_{u}, \lambda_u, \mu_u)$ as follows denote the Lagrangian of user $u$’s optimization problem $\text{OP}_u$:
\begin{align*}
    \mathscr{L}_u &(\textbf{p}_{u}, \lambda_u, \mu_u) = \sum^{k_u}_{i=1}p_{u,i} C_{u,i}(\textbf{p}_{u},\textbf{p}_{-u}) \\
    & + \lambda_u \left[\sum^{k_u}_{i=1} p_{u,i} \log \left(\frac{p_{u,i}}{p_{u,i}^o}\right)-T_u\right] - \mu_{u} \sum^{k_u}_{i=1} p_{u, i}, 
\end{align*} where $\lambda_u \in \mathbb{R}_{+}$, $\mu_{u} \in \mathbb{R}_{+}$ are the Lagrange multipliers.
We consider
$
    C_{u,i}(\textbf{p}_{u},\textbf{p}_{-u}) = \sum_{e \in s_{u, i}} t_e\big(1+\alpha\big(\frac{x_e^r(\textbf{p})}{k_e}\big)^\beta\big),
$ and the first-order condition $\partial \mathscr{L}_u/\partial p_{u, i} = 0$ for each $p_{u, i}$ becomes:
\begin{align*}
    \sum_{e \in s_{u, i}} & t_e \bigg[1+\alpha\bigg(\frac{x_e^r(\textbf{p})}{k_e}\bigg)^\beta + \beta x_e^u(\textbf{p}_u) \frac{\alpha}{k_e}\bigg(\frac{x_e^r(\textbf{p})}{k_e}\bigg)^{\beta-1}\bigg] \\
    & \qquad \qquad \quad + \lambda_u \left[\log\left(\frac{p_{u, i}}{p_{u, i}^o}\right) + 1\right]-\mu_{u}=0.
\end{align*} By letting $\log(\mu_{u}^\prime) = - \mu_{u}/\lambda_u$, then
\begin{align*}
    \frac{C^\prime_{u, i}(\textbf{p})}{\lambda_u} + \log\left(\frac{p_{u, i}}{p_{u, i}^o}\right) + 1 + \log(\mu_{u}^\prime)=0.
\end{align*}
Therefore, for all $u \in \mathcal{U}, s_{u, i} \in \mathcal{S}_u$, we have each
\begin{equation*}
    p_{u, i}^* =  \frac{p_{u, i}^o \exp\left(\frac{-C^\prime_{u, i}(\textbf{p})}{\lambda_u}-1\right)}{\mu_{u}^\prime}, 
\end{equation*} where $\mu_u, \forall u \in \mathcal{U}$ are normalizations ensuring $\textbf{p}^* \in \mathcal{P}$.
\end{proof}

Under stable conditions, for each \( s_{u, i} \in \mathcal{S}_u \), \( u \in \mathcal{U} \) being used, \( C^\prime_{u, i}(\textbf{p}) \) remains identical, and \(\textbf{p}\) remains the same as \(\textbf{p}^o\). However, if there is a demand attack, each \( C^\prime_{u, i}(\textbf{p}) \) perceived by the NRS will differ, causing \(\textbf{p}\) to deviate from \(\textbf{p}^o\) and tilt towards paths with lower perceived costs (which may not be the true costs) caused by misinformed demand. Additionally, the extent of deviation from previous \(\textbf{p}^o\) to current \(\textbf{p}\) depends on the multipliers \(\lambda_u\), which are associated with the trust score \(T_u\) for each \(u \in \mathcal{U}\). Let $\textbf{p}_u^*$ denote the trusted recommendation for user $u$ in Proposition \ref{prop:miti}, the optimal $\lambda_u^*$ can be found by numerically evaluating the dual function defined below. 
\begin{equation}
\begin{aligned}
    &\mathscr{G}_u(\lambda_u) := \min_{\textbf{p}_u \in \mathcal{P}_u} F_u^r(\textbf{p}_{u},\textbf{p}_{-u}) + \lambda_u \left[D(\textbf{p}_u||\textbf{p}_u^o)-T_u\right]\\
    &= \sum^{k_u}_{i=1}p_{u,i}^* C_{u,i}(\textbf{p}_{u}^*,\textbf{p}_{-u})
    + \lambda_u \big[\sum^{k_u}_{i=1} p_{u,i}^* \log \big(\frac{p_{u,i}^*}{p_{u,i}^o}\big)-T_u\big]
\end{aligned}
\end{equation}
which leads us to the dual problem of $\text{OP}_u$ as follows:
\begin{equation}
\begin{aligned}
    \text{DOP}_u: \max_{\lambda_u \in \mathbb{R}_{+}} \ & \mathscr{G}_u(\lambda_u)
\end{aligned}
\label{eq:DOP_u} 
\end{equation}

Note that Proposition \ref{prop:miti} considers the situation that the trust score $T_u$ is carefully determined so that TC is active with $\lambda_u >0$. When $\lambda_u=0$, which suggests that TC is non-binding, potentially due to the UE recommendation $\textbf{p}_u$ defined in Definition \ref{def:NRS} under current traffic condition is close to the previous experienced $\textbf{p}_u^o$, or because the user has high confidence in the current recommendation (i.e., $T_u$ is large), the NRS can then recommend the user with $\textbf{p}_u^* \in \argmin_{\textbf{p}_{u} \in \mathcal{P}_u} F_u^r(\textbf{p}_{u}, \textbf{p}_{-u})$.

To this end, the proposed trust mechanism can be summarized by the following Algorithm \ref{algo:miti}.
\begin{algorithm}
  \caption{Trust Mechanism}\label{algo:miti}
  \begin{algorithmic}[1]
    \State\textbf{Input} NRS component $\mathscr{R}=\left\langle   \mathcal{G}, (c_e(\cdot))_{e \in \mathcal{E}}, \mathcal{U}, (\mathcal{S}_u)_{u \in \mathcal{U}}  \right\rangle$
    \State\textbf{Collect} trust scores $T_u$ from all the users $u \in \mathcal{U}$
    \State\textbf{Obtain} $\textbf{p}_u^o, \forall u \in \mathcal{U}$ from historical data
    \State\textbf{Initialize} recommendation $\textbf{p}$ based on \eqref{eq:PGD}
    \For{$u \in \mathcal{U}$}
        \If{TC for $u$ non-binding}
            \State $\textbf{p}_u^*=\textbf{p}_u$
        \Else
            \State $p_{u, i}^* =  \frac{p_{u, i}^o \exp\left((-C^\prime_{u, i}(\textbf{p})/\lambda_u)-1\right)}{\mu_{u, i}^\prime}, \forall s_{u, i} \in \mathcal{S}_u$
            \State $\lambda_u^* \in \argmax_{\lambda_u \in \mathbb{R}_{+}} \mathscr{G}_u(\lambda_u)$
        \EndIf
    \EndFor
    \State \textbf{Return} $\textbf{p}^*$ to users and PRADA risk evaluator
  \end{algorithmic}
\end{algorithm}

In practice, recognizing the vulnerabilities illustrated in Fig. \ref{fig:NRS_process}, an NRS can consult the PRADA risk evaluator to assess risks for threat models from attack libraries. If the risk metrics TI and NI in the reports surpass the company's standards, one approach for the NRS to mitigate these risks is to collect user trust scores and implement a trust mechanism. The PRADA risk evaluator can then use the trust recommendation $\textbf{p}^*$ from Algorithm \ref{algo:miti} to reassess the risks and ensure they align with the company's standards.

\subsection{Sensitivity Analysis for Trust Mitigation}
In this subsection, we aim to examine the relationship between the multiplier $\lambda_u$, the optimal value for problem \eqref{prob:RS2}, and the user's trust score $T_u$. Suppose $T_u$ is changed to $T_u^\prime$ (due to positive or negative news related to the NRS), where $T_u^\prime = T_u + \eta_u$, the TC then becomes 
\begin{equation}
    D(\textbf{p}_u||\textbf{p}_u^o)-T_u \leq \eta_u.
\end{equation}

\begin{proposition}
    Consider the optimization problem $\text{OP}_u$. Under the assumptions that $F_u$ is convex in $\textbf{p}_u$ and $T_u, T_u^\prime$ are chosen so that the Slater's condition holds, let $v^*_{u, 0}$ and $v^*_{u, \eta_u}$ denote the optimal value for $\text{OP}_{u}$ associated with $T_u$, and $T_u^\prime$, respectively, and let $\lambda^*_u$ represent the optimal dual variable for $\text{OP}_u$ associated with $T_u$, then 
    \begin{equation}
        v^*_{u, t_u} \geq v^*_{u, 0} - \lambda^*_u \eta_u.
        \label{eq:opt_value}
    \end{equation}
\label{prop:tu}
\vspace{-3mm}
\end{proposition}
\begin{proof}
    Suppose that $\textbf{p}_u \in \mathcal{P}_u$ is any feasible point for $\text{OP}_u$ associated with $T_u^\prime$, then by strong duality
    \begin{align*}
        v^*_{u, 0} = \mathscr{G}_u(\lambda_u^*) &\leq F_u^r(\textbf{p}_{u},\textbf{p}_{-u}) + \lambda_u^*\left[D(\textbf{p}_u||\textbf{p}_u^o)-T_u\right]\\
        & \leq v^*_{u, \eta_u} + \lambda_u^* \eta_u,
    \end{align*} which completes the proof.
\end{proof}
Proposition \ref{prop:tu} indicates that if $\lambda^*_u$ is large and user $u$ shrinks his/her trust region (i.e., $\eta_u < 0$), then the optimal value of user $u$'s expected cost becomes much higher. Conversely, if $\lambda^*_u$ is small, even if user $u$ expands his/her trust region (i.e., $\eta_u > 0$), the optimal value of user $u$'s expected cost does not decrease substantially. 
That is, when $\lambda^*_u$ is large (small), the optimal value of user $u$'s expected cost is more (less) sensitive to changes in the trust region. Thus, $\lambda^*_u$ can be interpreted as a \textit{risk factor}. A larger $\lambda^*_u$ indicates a higher risk for user $u$ in modifying the trust region, as it leads to more significant changes in the user's optimal expected cost.

It is essential for users to understand that the trust score must be carefully determined due to the trade-offs between low and high values. For instance, if a user has low confidence in the NRS, resulting in a very small $T_u$, the recommendation will closely follow the user's past experiences. While this minimizes the impact of potential demand attacks, it also means the user may lose the chance to adapt to gradually changing traffic conditions if no attack occurs. On the other hand, a high trust score allows users to receive recommendations that reflect the latest traffic conditions, optimizing their travel time when there is no attack. However, this high trust also increases susceptibility to demand attacks, potentially leading to more significant manipulation of their recommendations.

\section{Discussion through Case Study}\label{sec:experiment}
We use the traffic network abstracted in Fig. \ref{fig:new_exp_1} as a case study of our PRADA framework, where we adopt the structure from the Sioux Falls network \cite{leblanc1975efficient}, and utilize the BPR function for the cost $c_e(\cdot)$ on each road $e \in \mathcal{E}$ with parameters $\alpha=0.4$, $\beta=2$, and $k_e=50$ for simplicity. The number displayed on each edge represents the free-flow time cost $t_e$. Then, we focus on the case where $20$ users seeking to travel from node $2$ to $17$ (OD $2$-$17$) and other $20$ users from node $9$ to $19$ (OD $9$-$19$). The feasible path set for users is specified in Table \ref{tab:setup1}. 

\begin{table}[htbp]
\caption{Case Study Setup}
\begin{center}
\begin{tabular}{ccl}
\toprule
OD pair & User \# & Feasible paths \\
\midrule
$2$-$17$ & $1$ to $20$ & \thead[l]{Path $1$: 2-6-8-9-10-16-17 \\ Path $2$: 2-6-8-9-10-17 \\ 
Path $3$: 2-6-8-16-17} \\
\midrule
$9$-$19$ & $21$ to $40$ & \thead[l]{Path $1$: 9-10-11-14-15-19 \\ Path $2$: 9-10-15-19 \\ 
Path $3$: 9-10-16-17-19 \\ 
Path $4$: 9-10-17-19} \\
\bottomrule \\[-0.3em]
\multicolumn{3}{c}{BPR cost with $\alpha=0.4$, $\beta=2$, $k_e=50, \forall e \in \mathcal{E}$.}
\end{tabular}
\end{center}
\label{tab:setup1}
\vspace{-5mm}
\end{table}

\subsection{Risks under Different Attacker Models}
In the context of a misinformed demand attack, attack methods (1)-(5) in subsection \ref{sec:att_methods} lead to fabricated user demands for a set $\mathcal{K} \subset \mathcal{V} \times \mathcal{V}$ of distinct OD pairs. We consider the case where the attacker has a local-targeted objective, and aims to make NRS recommend a level of $\gamma=20$ flow load from authentic users passing $(10, 17)$, the target road. (The flow load without attack is $12$, originally.) We compare the risk in terms of TI and NI of the following types of attackers in Fig. \ref{fig:exp1_}. This risk report provides the PRADA risk evaluator with a holistic overview, highlighting the attacks that require the most attention and urgent mitigation.

\begin{figure}
    \centering
    \includegraphics[width=2.9in]{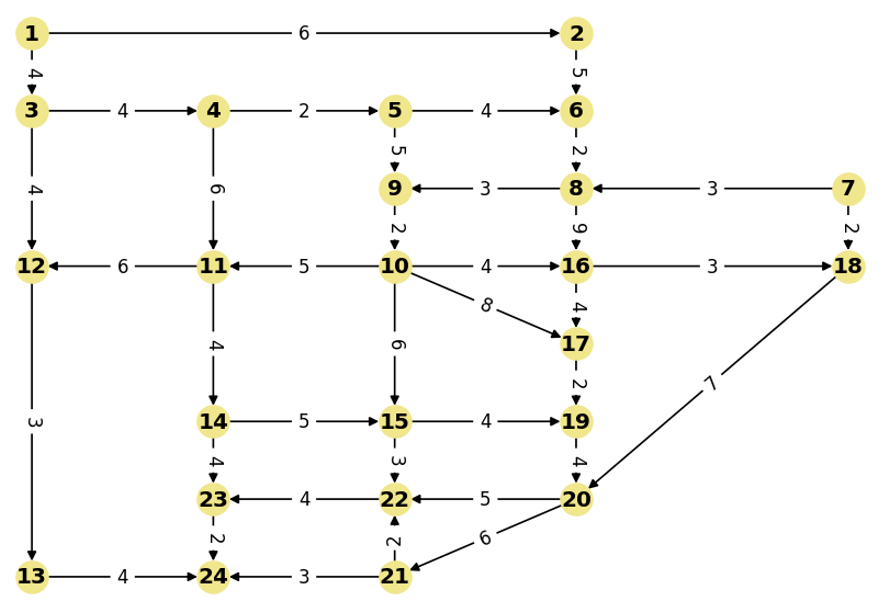}
    \caption{An example network (based on the network structure of Sioux Falls) for our case study. The value on each edge denotes the free-flow road travel time.}
    \label{fig:new_exp_1}
\vspace{-5mm}
\end{figure}

\begin{figure*}
    \centering
    \includegraphics[width=5.in]{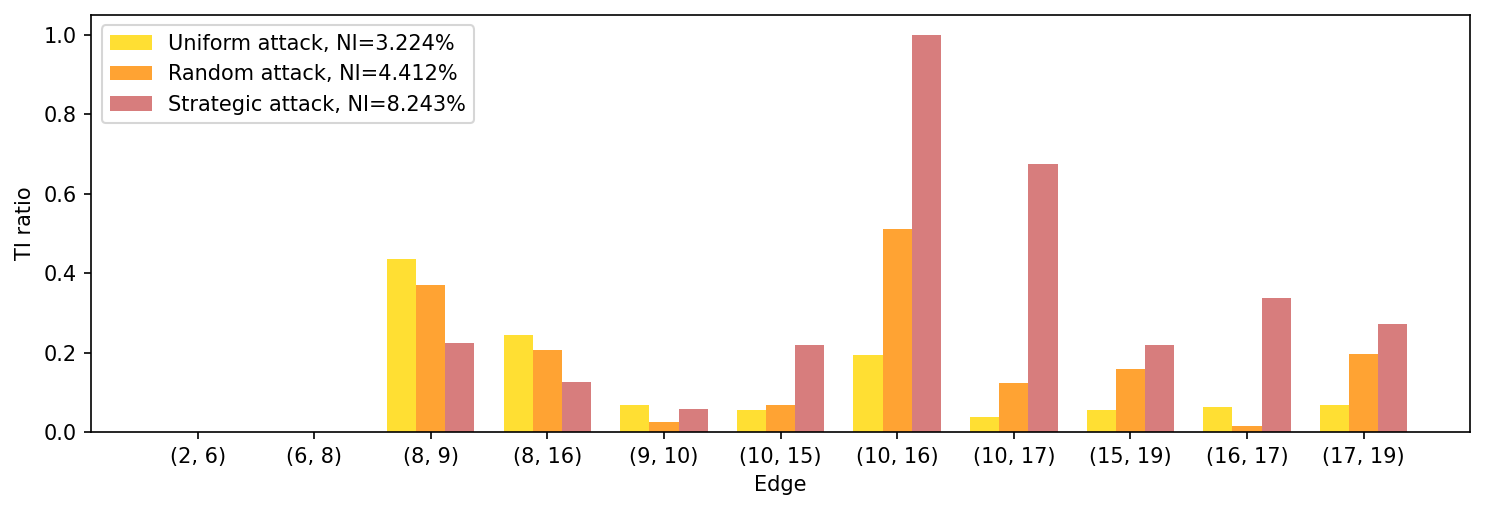}
    \caption{Risk report in terms of TI (local-targeted impact on roads along users' feasible paths) and NI (network-wide impact) when encountering non-strategic (random, uniform) and strategic attackers.}
    \label{fig:exp1_}
\vspace{-3mm}
\end{figure*}

\begin{figure*}
    \centering
    \includegraphics[width=5.in]{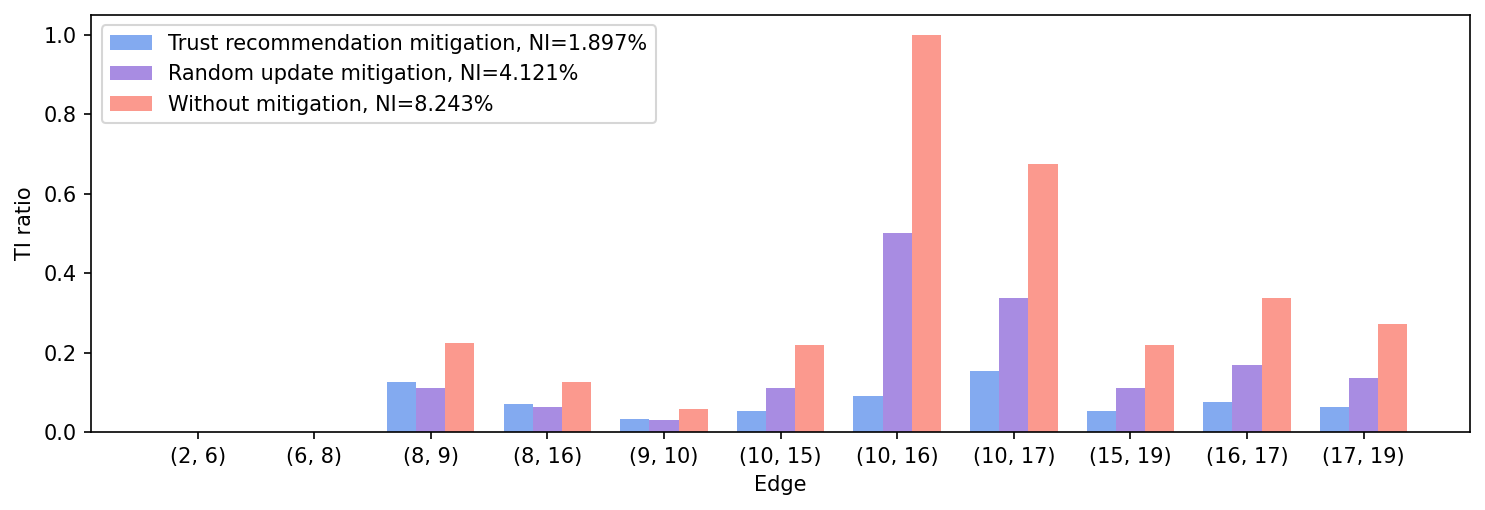}
    \caption{Risk report in terms of TI and NI when adopting mitigation methods (random update and trusted recommendation) under strategic attacks. The trusted recommendation mitigation can effectively reduce both the severe local-targeted and network-wide impacts compared to the random update one.}
    \label{fig:mitigate}
\vspace{-5mm}
\end{figure*}

\subsubsection{Strategic Attacker} A strategic attacker who has the knowledge of NRS can identify the desired fake demand levels by solving the leader-follower problem in section \ref{sec:md_problem}. The expected flow on the target road $(10, 17)$ caused by authentic users meets the desired level $\gamma \geq 20$ by generating a total of less than $35$ non-existent demands within the traffic network.

\subsubsection{Non-strategic Attacker}
A non-strategic attacker may not know how NRS generates the recommendation for users. Hence, the uniform attacker evenly distributes the total demand across all OD pairs near the target road, while the random attacker distributes the demand randomly among the OD pairs. Note that for comparison, the non-strategic attackers are restricted to allocating the same amount of demands, totaling $35$, as the optimal strategic attacker.

Fig. \ref{fig:exp1_} shows the risk report in terms of network-wide impact (NI) and local-targeted impact (TI) on roads/edges along users' feasible paths. First, we can observe that none of the attack scenarios affect roads $(2, 6)$ and $(6, 8)$. Since true users with the OD pair $2$-$17$ must travel through these two roads to reach their destination (as all three feasible paths include these roads), the attacker cannot impact these roads by redistributing users through fake demands. This suggests that the road users must go through are at lower risk of demand attack, as it is hard for malicious entities to influence the flow load by fabricating non-existent demands. Secondly, we see that all attack scenarios result in changes in traffic flow on roads other than $(2, 6)$ and $(6, 8)$, with the strategic attacks causing more severe impacts on most roads. Moreover, the alternative roads to the target road $(10, 17)$, including $(10, 15)$, $(10, 16)$, and $(16, 17)$, are at higher risk. This is because the attacker achieves their goal by manipulating the NRS to redistribute users originally passing through these roads to the target road, which also illustrates the analysis in Section \ref{sec:sensitivity}. Lastly, the network-wide impact (NI) indicates that the risks posed by strategic attackers are higher compared to non-strategic ones. This heightened risk in the risk report emphasizes the urgent need for mitigation strategies against intelligent attackers.

\subsection{Potential Mitigation of the Risk}
In this subsection, we aim to assess the risk when the NRS adopts the trusted recommendation described in Definition \ref{def:NRS2}. To evaluate whether such a trust mechanism can effectively mitigate the impact of misinformed demand attacks, we compare it with a straightforward random update mitigation method and a scenario without any mitigation. The risk report associated with these mitigation methods is shown in Fig. \ref{fig:mitigate}. This report aids the PRADA risk evaluator in determining the most efficient mitigation mechanism.

\subsubsection{Trusted Recommendation Mitigation} In practice, users may not adhere to recommendations that differ significantly from previously received recommendations for the same OD pair. Therefore, to ensure user compliance, the NRS incorporates user trust constraints into its recommendations, called trusted recommendations. Such a trust mechanism ensures that the current recommendation does not deviate significantly from the previous one for the same OD pair. The degree of deviation allowed in the current recommendation depends on the user's trust score $T_u$. A higher trust score indicates greater user trust in the integrity of NRS, with the user interpreting deviations as responses to sudden changes in traffic conditions rather than malicious demand attacks. In this context, the current recommendation for all users is given based on Algorithm \ref{algo:miti}.

\subsubsection{Random Update Mitigation}
In practice, not all users receive the updated recommendations simultaneously. Therefore, we consider the random update algorithm that may assist in mitigating the risk of sudden changes in demands perceived by the NRS caused by demand attacks. In this context, each user gets his/her current recommendation  $\textbf{p}_u$ with a predefined probability $0 < \pi_u < 1$; otherwise remains $\textbf{p}_u^o$. That is, 
\begin{equation}
    \textbf{p}_u = 
    \begin{cases}
\textbf{p}_u \ \text{satisfying \eqref{def:NRS}}, \quad \text{w.p.} \ \pi_u,\\
        \textbf{P}_u^o, \quad \qquad \qquad \text{w.p.}\ 1-\pi_u.
    \end{cases}
\label{eq:update_P_u_r}
\end{equation}
The probability of update $\pi_u$ may vary based on the user's driving habits or the capabilities of V2X technologies within the region containing the user's origin and destination. Here, we consider $\pi_u=0.5, \forall u \in \mathcal{U}$ for simplicity.

Fig. \ref{fig:mitigate} shows the risk report in terms of NI and TI on roads/edges along users' feasible paths when the NRS adopts mitigation methods. With $\pi_u = 0.5$ for all users $u \in \mathcal{U}$, the random update mitigation method reduces the risk from strategic attacks by half. It is important to note that a lower $\pi_u$ could potentially decrease risks further, but it may also result in users losing access to the most recent recommendations based on current traffic conditions.
As for the proposed trusted recommendation, it can effectively mitigate risk by constraining flow load changes based on user trust scores. Additionally, comparing road $(10, 16)$ with road $(8, 16)$, we can observe that the trusted recommendation mitigation is more obvious when the roads are originally facing higher risks of demand attack. 

\subsection{Discussion on the Resilience Paradox}

To this end, a natural question is: \textit{Can the locally targeted attack lead to a better overall outcome (total travel time costs for users) in some situations?} We can start with a carefully crafted example leveraging the classical Braess’ network \cite{Braess}. 

\begin{figure}[!h]
    \centering
    \begin{subfigure}[t]{0.21\textwidth}
        \centering
        \includegraphics[width=\textwidth]{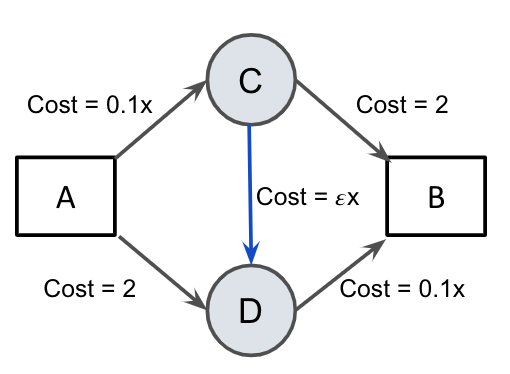}
        \caption{Without attack}
        \label{fig:paradox1}
    \end{subfigure}
    \hfill
    \begin{subfigure}[t]{0.21\textwidth}
        \centering \includegraphics[width=\textwidth]{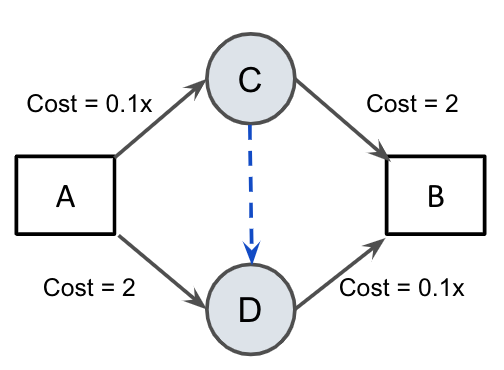}
        \caption{Under attack}
        \label{fig:paradox2}
    \end{subfigure}
    \caption{A carefully crafted example for the discussion on the ``Resilience Paradox''.}
\label{fig:paradox}
\vspace{-5mm}
\end{figure}

Within the transportation network shown in Fig. \ref{fig:paradox}, there are $30$ users aiming to go from node A to node B, and the $\epsilon$ is small enough so that the cost on C-D is close to $0$ even though all $30$ users are passing through. Before the attack (illustrated in Fig. \ref{fig:paradox1}), the RS will recommend a mixed strategy $(1/3, 1/3, 1/3)$ on path A-C-B, A-C-D-B, and A-D-B, respectively. The overall costs on these three paths are all $4$, which leads to a total travel time cost of $120$ for users. Suppose the attacker wants more ``users'' to pass D-B by fabricating a large demand on C-D to make C-D seem congested to the RS, as in Fig. \ref{fig:paradox2}. The RS will recommend a strategy $(0.5, 0, 0.5)$ on paths A-C-B, A-C-D-B, and A-D-B, respectively. The overall costs on A-C-B and A-D-B are both $3.5$, which leads to the total travel time cost for users becoming $105$. Therefore, we can conclude that the cost under attack is better than the performance without attack in this carefully crafted example.

\section{Conclusion}
This paper assesses the risk of potential informational attacks on navigational recommendation systems (NRS). We introduce the attack model and identify vulnerabilities that attackers can exploit to launch demand attacks, achieving locally targeted goals that benefit certain groups or businesses. Then, we propose a holistic framework for proactive risk assessment of demand attacks (PRADA)  that integrates necessary models.
Given that modern attackers are often intelligent, our focus, from the perspective of the PRADA risk evaluator, is on strategic attacks. We analyze the interaction between the attacker and the incentive-compatible NRS through a Stackelberg game. Our study indicates that users are at high risk when facing strategic attacks that target specific roads by creating non-existent demands for OD pairs with alternative path options. To mitigate these risks, we introduce a trust mechanism, and our investigation shows that it is a viable approach to reducing the risk posed by misinformed demand attacks in both local-targeted and network-wide senses.


\bibliography{reference}
\bibliographystyle{IEEEtran}

\end{document}